\documentclass[12pt]{article}

\usepackage{amsfonts}
\usepackage{amssymb}
\usepackage{amsmath}
\usepackage{amsthm,color}

\usepackage{verbatim}
\usepackage{hyperref}
\allowdisplaybreaks

\usepackage{setspace}

\usepackage[top=1in, bottom=1in, left=1in, right=1in]{geometry}

\newtheorem{theorem}{Theorem}
\newtheorem{corollary}[theorem]{Corollary}
\newtheorem{rem1}[theorem]{Remark}
\newtheorem{lemma}[theorem]{Lemma}
\newtheorem{definition}[theorem]{Definition}
\newtheorem{proposition}[theorem]{Proposition}

\newtheorem{ex1}[theorem]{Example}
\newtheorem{ass1}[theorem]{Assumption}

\newenvironment{remark}{\begin{rem1}\rm}{\end{rem1}}

\newenvironment{assumption}{\begin{ass1}\rm}{\end{ass1}}

\numberwithin{equation}{section}
\numberwithin{theorem}{section}

\newcommand{\N}{\mathbb{N}}
\renewcommand{\P}{\mathbb{P}}
\newcommand{\Q}{\mathbb{Q}}
\newcommand{\R}{\mathbb{R}}
\renewcommand{\S}{\mathbb{S}}
\newcommand{\W}{\mathcal{W}}

\newcommand{\tildeW}{\widetilde{\W}}

\newcommand{\K}{\mathbb{K}}

\newcommand{\lrparen}[1]{\left(#1\right)}
\newcommand{\lparen}[1]{\left(#1\right.}
\newcommand{\rparen}[1]{\left.#1\right)}
\newcommand{\lrsquare}[1]{\left[#1\right]}

\newcommand{\lrcurly}[1]{\left\{#1\right\}}

\newcommand{\Ft}[1]{\mathcal{F}_{#1}}
\newcommand{\Lnp}[3]{L_{#1}^{#2}(#3)}

\newcommand{\LdpF}[1]{\Lnp{#1}{p}{\R^d}}
\newcommand{\LdqF}[1]{\Lnp{#1}{q}{\R^d}}

\newcommand{\LdiF}[1]{\Lnp{#1}{\infty}{\R^d}}

\newcommand{\LdzF}[1]{\Lnp{#1}{0}{\R^d}}
\newcommand{\LdpK}[3]{\Lnp{#2}{#1}{#3}}
\newcommand{\LdiK}[2]{\Lnp{#1}{\infty}{#2}}
\newcommand{\LdoK}[2]{\Lnp{#1}{1}{#2}}

\newcommand{\LpF}[1]{\Lnp{#1}{p}{\R}}
\newcommand{\LqF}[1]{\Lnp{#1}{q}{\R}}
\newcommand{\LiF}[1]{\Lnp{#1}{\infty}{\R}}
\newcommand{\LoF}[1]{\Lnp{#1}{1}{\R}}

\newcommand{\LpK}[3]{\Lnp{#2}{#1}{#3}}
\newcommand{\barLzF}[1]{\LpK{0}{#1}{\R \cup \{\pm\infty\}}}

\newcommand{\E}[1]{\mathbb{E}\lrsquare{#1}}
\newcommand{\EP}[2]{\mathbb{E}^{#1}\lrsquare{#2}}

\newcommand{\Et}[2]{\E{\left.#1 \right| \mathcal{F}_{#2}}}
\newcommand{\EPt}[3]{\EP{#1}{\left.#2 \right| \mathcal{F}_{#3}}}
\newcommand{\EQt}[2]{\EPt{\Q}{#1}{#2}}

\newcommand{\EQnt}[3]{\EPt{\Q_{#1}}{#2}{#3}}
\newcommand{\ERt}[2]{\EPt{\R}{#1}{#2}}

\newcommand{\dQdP}{\frac{d\Q}{d\P}}

\newcommand{\dRdP}{\frac{d\R}{d\P}}

\newcommand{\genseq}[4]{\lrparen{#1_#4}_{#4=#2}^{#3}}
\newcommand{\seq}[1]{\genseq{#1}{0}{T}{t}}

\newcommand{\trans}[1]{#1^{\mathsf{T}}}
\newcommand{\transp}[1]{\trans{\lrparen{#1}}}
\newcommand{\prp}[1]{#1^{\perp}}

\newcommand{\diag}[1]{\operatorname{diag}\lrparen{#1}}
\newcommand{\cl}{\operatorname{cl}}
\newcommand{\co}{\operatorname{co}}

\newcommand{\abs}[1]{\left|#1\right|}

\newcommand{\Pas}{\P\text{-a.s.}}

\DeclareMathOperator*{\esssup}{ess\,sup}
\DeclareMathOperator*{\essmax}{ess\,max}
\DeclareMathOperator*{\essinf}{ess\,inf}

\DeclareMathOperator*{\argessmax}{arg\,ess\,max}

\begin{document}
\title{Time consistency for scalar multivariate risk measures}
\author{Zachary Feinstein \thanks{Stevens Institute of Technology, School of Business, Hoboken, NJ 07030, USA, zfeinste@stevens.edu.} \and Birgit Rudloff \thanks{Vienna University of Economics and Business, Institute for Statistics and Mathematics, Vienna A-1020, AUT, brudloff@wu.ac.at. B. Rudloff acknowledges support from the OeNB anniversary fund, project number 17793.}}
\maketitle
\abstract{
In this paper we present results on dynamic multivariate scalar risk measures, which arise in markets with transaction costs and systemic risk.  Dual representations of such risk measures are presented.  These are then used to obtain the main results of this paper on time consistency; namely, an equivalent recursive formulation of multivariate scalar risk measures to multiportfolio time consistency.  We are motivated to study time consistency of multivariate scalar risk measures as the superhedging risk measure in markets with transaction costs (with a single eligible asset) \cite{JK95,RZ11,LR11} does not satisfy the usual scalar concept of time consistency.  In fact, as demonstrated in \cite{FR18-1asset}, scalar risk measures with the same scalarization weight at all times would not be time consistent in general.  The deduced recursive relation for the scalarizations of multiportfolio time consistent set-valued risk measures provided in this paper requires consideration of the entire family of scalarizations.  In this way we develop a direct notion of a ``moving scalarization'' for scalar time consistency that corroborates recent research on scalarizations of dynamic multi-objective problems \cite{KMZ17,KR18}.
}

\section{Introduction}
\label{sec:intro}

The organization of this paper is as follows.  First, we complete this introduction with a review of relevant literature (Section~\ref{sec:lit}) and a detailed motivation for our study (Section~\ref{sec:mot}). In Section~\ref{sec:prelim} we present the mathematical setting we will utilize throughout this paper.  Much of this notation is comparable to that from the set-valued risk measure literature.  In Section~\ref{sec:setvalued} we summarize basic definitions related to set-valued risk measures which will be used later in this work.  We then consider the scalarizations in Section~\ref{sec:scalar}.  Many of the results on the primal and dual representation of multivariate scalar risk measures come from~\cite{FR13-survey,FR15-supermtg}, but they are repeated here for completeness of this work.  The main results of this paper are found in Section~\ref{sec:mptc}.  In that section we present a recursive formulation on the family of scalarized risk measures that is equivalent to multiportfolio time consistency.  In particular, we wish to highlight that this recursive formulation naturally provides the notion of a ``moving scalarization'' for time consistency as discussed in \cite{KR18} for the mean-risk portfolio optimization problem.  By considering multiportfolio time consistency, we can immediately demonstrate that this recursive relation is satisfied by the usual examples; in particular we give details on the superhedging risk measure and composed risk measures in Section~\ref{sec:examples}. Most proofs are relegated to the Appendix. Furthermore, an auxiliary dual representation for the multivariate scalar risk measures, used extensively in the proofs of the main results, is provided in Appendix~\ref{sec:dual-aux}.

\subsection{Literature review}\label{sec:lit}
Coherent risk measures were introduced in an axiomatic manner in the seminal paper by Artzner, Delbaen, Eber, and Heath \cite{AD99}.  They work in a static and univariate setting in which a contingent claim is defined by its value in some num\'eraire and the risk measure outputs the minimal capital necessary to compensate for the risk of the initial claim.  Convex risk measures were introduced as a generalization of the coherent risk measures in \cite{FS02,FG02} for the static univariate setting; such risk measures retain the same interpretation and the notion that diversification does not increase risk.  

With the introduction of time and information, i.e.\ with the introduction of a filtration $\seq{\mathcal{F}}$, it is natural to consider dynamic risk measures.  In a univariate setting, such functions provide the minimal capital required to compensate for the risk of a portfolio conditional on the information available at time $t$.  With the introduction of a filtration and time dynamics, it is natural and vital to consider the manner in which the risk of a portfolio or contingent claim propagates in time.  This propagation of risk over times has significant implications on risk management.  One condition for how risks propagate in time is called (strong) time consistency; it consists of the condition that if one claim is riskier than another in the future, that claim should be riskier for all prior times as well.  Time consistency for dynamic univariate risk measures has been studied in~\cite{AD07,R04,DS05,CDK06,RS05,BN04,BN08,BN09,FP06,CS09,CK10,AP10,FS11} in discrete time and~\cite{FG04,D06,DPRG10} in continuous time.

In this work, we focus on multivariate risk measures.  This setting has been studied in two contexts in the literature.  It was originally introduced for studying markets with frictions; in such a market the liquidation of a portfolio or claim is not reversible, i.e.\ the liquidation value is not enough to repurchase the same asset.  Risk measures for markets with transaction costs were studied in a static setting with a set-valued framework in~\cite{JMT04,HHH07,HR08,HH10,HHR10}.  In a dynamic framework for markets with transaction costs, time consistency of set-valued risk measures was introduced and studied in~\cite{FR12,TL12}.  This notion of time consistency is called multiportfolio time consistency; this property has been studied in~\cite{FR12b,FR13-survey,FR15-supermtg,CH17}.  The computation of static and dynamic set-valued risk measures was considered in~\cite{FR14-alg,LR11}.  The second context in which multivariate risk measures have recently been considered is for studying systemic risk.  Systemic risk measures were defined in a set-valued framework in~\cite{FRW15,BFFM15} with dual representations deduced in~\cite{ararat2020dual,arduca2019dual}. 

The focus of this work is on dynamic scalar multivariate risk measures with a general number of eligible assets.  We refer to \cite{FR13-survey,FR15-supermtg} for some prior results on dynamic scalar multivariate risk measures with multiple eligible assets.  Where appropriate, we will summarize those results in the body of this paper as well.
The special case of multivariate risk measures with a single eligible asset has been studied in~\cite{FR18-1asset}.
We note, as provided in \cite[Example~2.26]{FR13-survey}, that scalar risk measures in frictionless markets with either a single eligible asset \cite{FS02,AD99} or multiple eligible assets \cite{FMM13,BMM18,WH14,ADM09,FS06,K09,Sc04} can also be considered as scalar multivariate risk measures.
Scalar multivariate functions have been considered in the context of systemic risk in, e.g., \cite{KOZ16,BC13}.
Our focus in this paper is of the dynamic version of the multivariate risk measures with multiple eligible assets which are intimately related to the efficient allocation rules of \cite{FRW15} and systemic risk measures of \cite{BFFM15}.  Though the set-valued risk measures provide the full risk profile, these scalarizations are typically what a risk manager or regulator would report as they describe the optimal capital allocations.

\subsection{Motivation}\label{sec:mot}
From the financial literature, we are specifically motivated to study the problem of scalar multivariate risk measures with multiple eligible assets from the relation they have with systemic risk measures. 
As detailed in \cite{FRW15,BFFM15,ACDP15}, systemic risk measures are naturally multivariate with capital allocations for each bank in the financial system being dependent on the others.  It is, thus, natural to consider finding the collection of capital allocations for the firms that require the minimal (weighted) total capital.  This aggregate capital, as defined by the efficient allocation rules of \cite{FRW15}, is exactly a scalar multivariate risk measure.  Though these works have considered multivariate systemic risk measures in a static setting only, we are motivated by the consideration of systemic risk in a dynamic setting.

Furthermore, we wish to study a time consistency property that is satisfied by the typical examples, e.g.\ the superhedging risk measure.  This is in contrast to many prior works on time consistency for scalar multivariate measures~\cite{HMBS16,KOZ15} whose time consistency property does not include the superhedging risk measure as an example.  
We refer the reader to~\cite{bielecki2017survey} for a survey of other, weaker, notions of time consistency.

Though a notion of time consistency exists for univariate risk measures, a number of recent works on scalarizations of multi-objective problems \cite{KMZ17,KR18,FR18-1asset} indicate that this original notion of time consistency is too strong. We also refer to~\cite{weber2006distribution,roorda2007time,tutsch2008update} which, similarly, indicate the original notion of time consistency is often too strong for certain other classes of scalar risk measures.  In fact, the superhedging price is a typical example for time consistency from the literature and \cite{FR18-1asset} demonstrates that the univariate definition of time consistency is not satisfied for such a risk measure in markets with frictions.  However, multiportfolio time consistency, which was introduced for set-valued multivariate risk measures \cite{FR12,TL12}, holds for many examples of set-valued risk measures including the superhedging risk measure.  We are thus motivated to ask if an equivalent property of multiportfolio time consistency exists for the scalarized multivariate risk measures.  If so, we propose that such a property would provide a more appropriate concept of time consistency for scalar multivariate risk measures. We will show that an equivalent property exists.  In fact, this equivalent property is a recursive formulation for the scalar multivariate risk measures which requires the entire family of scalarizations.  To the best of our knowledge, this recursive formulation utilizing the entire family of scalarizations has never before been proposed or studied.  As such, these new results provide a foundation to study the ``moving scalarizations'' for (scalar) dynamic programs which allow for Bellman's principle to hold, e.g.\ for the dynamic mean-risk portfolio optimization problem~\cite{KMZ17,KR18}.

As highlighted in our literature review, time consistency of multivariate risk measures with a single eligible asset was studied in~\cite{FR18-1asset}.  In that single eligible asset setting, multiportfolio time consistency generally does \emph{not} hold and thus a weaker time consistency property was proposed and studied.  
The starting point of the present paper is different as we allow for multiple eligible assets. Now assuming multiportfolio time consistency for the underlying set-valued risk measure, we study its equivalent characterization for the scalarizations.

\section{Preliminaries}
\label{sec:prelim}
Throughout this work we will consider the filtered probability space $\lrparen{\Omega,\Ft{},\seq{\mathcal{F}},\P}$ in either discrete or continuous time.  Assume that this filtered probability space satisfies the usual conditions with $\Ft{0}$ the trivial sigma algebra and $\Ft{} = \Ft{T}$. Define $|\cdot|_n$ to be an arbitrary norm in $\R^n$ for $n \in \N$.   As such, the mapping $|\cdot|_1$ is the absolute value operator.
Throughout this work we will denote the set of equivalence classes of $\Ft{t}$-measurable functions taking value in the (random) set $D$ by $\LdpK{0}{t}{D} := \LdpK{0}{}{\Omega,\Ft{t},\P;D}$.  That is, $X \in \LdpK{0}{t}{D}$ if $X: \omega \in \Omega \mapsto D[\omega] \subseteq \R^n$ for some $n \in \N$ is $\Ft{t}$-measurable.  If $D \subseteq \R^n$ is constant, the set $\LdpK{0}{t}{D}$ is defined to be equivalent to $\LdpK{0}{t}{\tilde D}$ where $\tilde D[\omega] := D$ for every $\omega \in \Omega$.  
Further, for $p \in [1,+\infty]$, we will define $\LdpK{p}{t}{D} \subseteq \LdpK{0}{t}{D}$ to be those random variables $X \in \LdpK{0}{t}{D}$ with finite $p$-norm, i.e.\ such that $\|X\|_p = \lrparen{\int_\Omega |X(\omega)|_n^p d\P}^{\frac{1}{p}} < +\infty$ for $p \in [1,+\infty)$ and $\|X\|_{\infty} = \esssup_{\omega \in \Omega} |X(\omega)|_n < + \infty$ for $p = +\infty$.  We denote $\LdpK{p}{}{D} := \LdpK{p}{T}{D}$ for any $p \in \{0\} \cup [1,+\infty]$.  

Often in this paper we wish to consider the space of $\Ft{t}$-measurable random vectors $\LdzF{t}$ and the norm bounded random vectors $\LdpF{t}$.  As defined above, $\LdzF{t}$ is the linear space of equivalence classes of $\Ft{t}$-measurable mappings $X: \Omega \to \R^d$.  For $p \in [1,+\infty]$, $\LdpF{t} \subseteq \LdzF{t}$ denotes the space of the $\Ft{t}$-measurable random vectors $X: \Omega \to \R^d$ such that $\|X\|_p < +\infty$.
In fact for any choice $p \in \{0\} \cup [1,+\infty]$, a $d$-dimensional random vector $X$ is an element of $\LdpF{t}$ if and only if its components $X_1,...,X_d$ are elements of $\LpF{t}$.

Throughout this work, we will restrict ourselves to $p \in [1,+\infty]$.  For $p \in [1,+\infty)$, we will endow $\LdpF{t}$ with the norm topology (respectively $\LpF{t}$ with the norm topology) for all times $t$.  For $p = +\infty$, we will endow $\LdiF{t}$ with the weak* topology (respectively $\LiF{t}$ with the weak* topology) for all times $t$.  In this way we can consider the associated dual space $\LdqF{t}$ (respectively $\LqF{t}$) defined by $\frac{1}{p} + \frac{1}{q} = 1$ for any choice of $p \in [1,+\infty]$.  For the remainder of this work we will define $q \in [1,+\infty]$ so that it denotes the dual space associated with $p$.

Additionally, as we will consider two specific operators throughout this work, we wish to define both at this time.  First, the indicator function for some $D \in \Ft{}$ is denoted by $1_D: \Omega \to \{0,1\}$ and defined as
\[1_D(\omega) = \begin{cases}1 &\text{if } \omega \in D\\ 0 &\text{if } \omega \not\in D\end{cases}.\]
Second, the summation of sets is provided by the Minkowski addition, i.e.\ $A + B = \{a + b \; | \; a \in A, \; b \in B\}$ for any sets $A,B$ in the same linear space.

As detailed in the introduction, we are motivated to study multivariate risk measures by markets with transaction costs as well as systemic risk.  The ``portfolios'' have slightly different meaning in these two settings:
\begin{itemize}
\item In markets with transaction costs, consider a market with $d$ assets.  Portfolios are denoted in their ``physical units'' in each of the $d$ assets.  Due to the transaction costs, any fixed num\'eraire would be unable to uniquely define a portfolio (e.g., mark-to-market or liquidation value).  We refer to \cite{K99,S04,KS09} for prior works that discuss the use of physical portfolios.  Specifically, a portfolio $X \in \LdpF{t}$ has $X_i$ units of asset $i$ at time $t$.
\item In a systemic risk setting, consider a system with $d$ banks or financial firms. Portfolios are denoted by the wealth of each firm due to their primary business activities as a vector.  Specifically, a portfolio $X \in \LdpF{t}$ means that firm $i$ has wealth $X_i$ at time $t$.
\end{itemize}

Consider a linear subspace $M \subseteq \R^d$ which denotes the set of portfolios which can be used to compensate for the risk of a portfolio.  We will call this space the ``eligible portfolios''. By construction $M_t := \LdpK{p}{t}{M}$  is a closed linear subspace of $\LdpF{t}$, see Section~5.4 and Proposition~5.5.1 in \cite{KS09}.  
We will denote $M_+ := M \cap \R^d_+$ to be the nonnegative elements of $M$.  We will assume that $M_+ \neq \{0\}$, i.e.\ $M_+$ is nontrivial.  We will denote $M_{t,+} := M_t \cap \LdpK{p}{t}{\R^d_+} = \LdpK{p}{t}{M_+}$ to be the nonnegative elements of $M_t$ and $M_{t,-} := -M_{t,+}$ to be the nonpositive elements of $M_t$.

\begin{assumption}
\label{ass:M}
Throughout this work we will assume that only $m \in \{1,...,d\}$ of the ``assets'', and no other portfolios, are eligible.  In the setting of markets with transaction costs, this corresponds to the setting of a few reserve currencies, e.g.\ US Dollars, Euros, and Yen, are eligible assets.  From a systemic risk perspective, this corresponds to the setting in which systemic capital requirements are only placed on those firms deemed systemically important.  Without loss of generality, we will assume that these are the first $m$ assets, i.e., $M = \R^m \times \{0\}^{d-m}$.
\end{assumption}

In considering set-valued functions, the image space of interest is the set of upper sets which are denoted by $\mathcal{P}\lrparen{M_t;M_{t,+}}$.  Here the space of upper sets is defined as $\mathcal{P}\lrparen{\mathcal{Z};C} := \lrcurly{D \subseteq \mathcal{Z} \; | \; D = D + C}$ for some vector space $\mathcal{Z}$ and an ordering cone $C \subset \mathcal{Z}$.  Additionally, often the space of interest is the set of upper closed convex sets, which are denoted by $\mathcal{G}(\mathcal{Z};C) := \lrcurly{D \subseteq \mathcal{Z} \; | \; D = \cl\co\lrparen{D + C}} \subseteq \mathcal{P}(\mathcal{Z};C)$.

We conclude this section by describing the set of dual variables that will be utilized throughout this work.  This construction comes from the set-valued biconjugation theory as it was used in \cite{FR12,FR12b}.
Define $\mathcal{M}$ to be the space of probability measures absolutely continuous with respect to $\P$.  The space of $d$-dimensional probability measures absolutely continuous with respect to $\P$ is thus denoted by $\mathcal{M}^d$.
For any $\Q \in \mathcal{M}$, we will consider the $\P$-almost sure version of the $\Q$-conditional expectation.  This version of the conditional expectation is defined in, e.g., \cite{CK10,FR12}.  Briefly, define the conditional expectation for any $X \in \LpF{}$ by
\[\EQt{X}{t} := \Et{\bar{\xi}_{t,T}(\Q)X}{t}\]
where 
\[\bar{\xi}_{s,\sigma}(\Q)[\omega] := \begin{cases}\frac{\Et{\dQdP}{\sigma}(\omega)}{\Et{\dQdP}{s}(\omega)} & \text{if } \Et{\dQdP}{s}(\omega) > 0 \\ 1 &\text{else} \end{cases}\]
for every $\omega \in \Omega$.  For vector-valued probability measures $\Q \in \mathcal{M}^d$, the $\Q$-conditional expectation is defined component-wise, i.e., $\EQt{X}{t} = \transp{\EQnt{1}{X_1}{t},\dots,\EQnt{d}{X_d}{t}}$ for any $X \in \LdpF{}$. 
By construction, for any $\Q \in \mathcal{M}$ and any times $0 \leq t \leq s \leq \sigma \leq T$, it holds that $\dQdP = \bar{\xi}_{0,T}(\Q)$ and $\bar{\xi}_{t,\sigma}(\Q) = \bar{\xi}_{t,s}(\Q) \bar{\xi}_{s,\sigma}(\Q)$ almost surely.

From this space of probability measures, we can define the set of dual variables at time $t$ or stepped from time $t$ to $s$ from the set-valued biconjugation theory \cite{H09}.  In fact, we will provide a slightly more restrictive set of dual variables than previously considered so that the conditions are all taken almost surely.  As such the set of dual variables are given by
\begin{align*}
\W_t &:= \lrcurly{(\Q,w) \in \mathcal{M}^d \times \LdpK{q}{t}{M_+^+ \backslash \prp{M}} \; | \; w_t^T(\Q,w) \in \LdpK{q}{}{\R^d_+}},\\
\W_{t,s} &:= \lrcurly{(\Q,w) \in \mathcal{M}^d \times \LdpK{q}{t}{M_+^+ \backslash \prp{M}} \; | \; w_t^s(\Q,w) \in M_{s,+}^+} \supseteq \W_t.
\end{align*}
Above, and throughout this work, we define \[w_t^s(\Q,w) = \diag{w}\xi_{t,s}(\Q)\] for any times $0 \leq t \leq s \leq T$, where $\xi_{s,\sigma}(\Q) := \transp{\bar{\xi}_{s,\sigma}(\Q_1),\dots,\bar{\xi}_{s,\sigma}(\Q_d)}$, and $\diag{w}$ denotes the diagonal matrix with the components of $w$ on the main diagonal.
Additionally, we employ the notation $\prp{M} = \lrcurly{v \in \R^d \; | \; \trans{v}u = 0 \; \forall u \in M}$ to denote the orthogonal space of $M$, $\prp{M_t} = \LdpK{q}{t}{\prp{M}}$ to denote the orthogonal space of $M_t$, and $C^+=\lrcurly{v \in \mathcal{Z}^* \; | \; \langle v , u \rangle \geq 0 \; \forall u \in C}$ to denote the positive dual cone of a cone $C\subseteq \mathcal{Z}$ where $\mathcal{Z}^*$ is the topological dual space and $\langle \cdot,\cdot \rangle$ is the bilinear operator.

We note that, in the following results for a conditional risk measure at time $t$, we can consider those probability measures $\Q \in \mathcal{M}^d$ that are equal to $\P$ on $\Ft{t}$, i.e., $\Q_i(D) = \P(D)$ for every $D \in \Ft{t}$ and $i \in \{1,...,d\}$.  This is due to the construction of the $\P$-almost sure version of the $\Q$-conditional expectation.  However, for ease of notation we will refer only to the space of absolutely continuous probability measures.

\section{Set-valued risk measures}
\label{sec:setvalued}

In this section,  we provide a summary of basic definitions related to set-valued risk measures which will be used within this work. 

The set-optimization approach to dynamic risk measures is studied in \cite{FR12,FR12b}, where set-valued risk measures \cite{HH10,HHR10} were extended to the dynamic case.
A benefit of this method is that dual representations (see Theorem~4.7 in~\cite{FR12}) are obtained by a direct application of
the set-valued duality developed in~\cite{H09}, which allowed for the first time to study not only conditional coherent, but also convex set-valued risk measures.

In this setting we consider risk measures that map a portfolio vector into the
complete lattice $\mathcal{P}\lrparen{M_t;M_{t,+}}$ of upper sets.  The intuition behind this is that if an eligible portfolio, i.e., an element of
$M_t$, covers the risk of some random vector, then any almost surely larger eligible portfolio will cover the risk as well.

Set-valued conditional risk measures have been defined in \cite{FR12}.  Note that we give a stronger property for finiteness at zero than that given in \cite{FR12}. 
\begin{definition}
\label{defn_conditional}
A \textbf{\emph{conditional risk measure}} is a mapping $R_t: \LdpF{} \to \mathcal{P}(M_t;M_{t,+})$ which satisfies:
\begin{enumerate}
\item $\LdpK{p}{}{\R^d_+}$-monotonicity: if $Y - X \in \LdpK{p}{}{\R^d_+}$ then $R_t(Y) \supseteq R_t(X)$;
\item $M_t$-translativity: $R_t(X+m) = R_t(X)-m$ for any $X \in \LdpF{}$ and $m \in M_t$;
\item finiteness at zero: $R_t(0) \neq\emptyset$ is closed and $\Ft{t}$-decomposable, and $\P(\tilde R_t(0) = M) = 0$ where $\tilde R_t(0)$ is a $\Ft{t}$-measurable random set such that $R_t(0) = \LdpK{p}{t}{\tilde R_t(0)}$.
\end{enumerate}
A conditional risk measure is
\begin{itemize}
\item \textbf{\emph{normalized}} if $R_t(X) = R_t(X) + R_t(0)$ for
every $X \in \LdpF{}$;
\item\textbf{\emph{$\K$-compatible}} for some convex cone $\K \subseteq \LdpF{}$ if $R_t(X) = \bigcup_{K \in \K} R_t(X- K)$;
\item \textbf{\emph{local}} if for every $D \in \Ft{t}$ and every $X,Y \in \LdpF{}$, $R_t(1_D X + 1_{D^c} Y) = 1_D R_t(X) + 1_{D^c} R_t(Y)$;
\item\textbf{\emph{conditionally convex}} if for all $X,Y \in \LdpF{}$ and any $\lambda \in \LpK{\infty}{t}{[0,1]}$ it holds
$R_t(\lambda X + (1-\lambda)Y) \supseteq \lambda R_t(X) + (1-\lambda) R_t(Y)$;
\item\textbf{\emph{conditionally positive homogeneous}} if for all $X \in \LdpF{}$ and any $\lambda \in \LpK{\infty}{t}{\R_{++}}$ it holds
$R_t(\lambda X) = \lambda R_t(X)$;
\item\textbf{\emph{conditionally coherent}} if it is conditionally convex and conditionally positive homogeneous;
\item \textbf{\emph{closed}} if $\operatorname{graph}(R_t) = \lrcurly{(X,u) \in \LdpF{} \times M_t \; | \; u \in R_t(X)}$, i.e., the graph of the risk measure, is closed in the product topology;
\item\textbf{\emph{convex upper continuous (c.u.c.)}} if $R_t^-[D] := \lrcurly{X \in \LdpF{} \; | \; R_t(X) \cap D \neq \emptyset}$
is closed for any closed set $D \in \mathcal{G}(M_t;M_{t,-})$. 
\item\textbf{\emph{halfspace lower continuous (h.l.c.)}} if  $R_t^+[D] := \lrcurly{X \in \LdpF{} \; | \; R_t(X) \subseteq D}$
is closed for any closed halfspace $D \in \mathcal{G}(M_t;M_{t,+})$.
\end{itemize}
A \textbf{\emph{dynamic risk measure}} is a sequence $\seq{R}$ of conditional risk measures.  A dynamic
risk measure is said to have a certain property if $R_t$ has that property for all times $t$.
\end{definition}
The existence of the random set $\tilde R_t(0)$ introduced above is guaranteed by \cite[Theorem 2.1.6]{M05}.  This notation will be used for the remainder of this paper.
A static risk measure in the sense of~\cite{HHR10} is a conditional risk measure at time $0$.  Further, we wish to note that any conditionally convex risk measure is local.

In~\cite{FR12}, a primal representation for conditional risk measures is given via acceptance sets as in Definition~\ref{defn_acceptance}.
\begin{definition}
\label{defn_acceptance}
A set $A_t \subseteq \LdpF{}$ is a \textbf{\emph{conditional acceptance set}} at time $t$ if it satisfies $A_t + \LdpK{p}{}{\R^d_+} \subseteq A_t$,
$M_t \cap A_t \neq \emptyset$ is closed and $\Ft{t}$-decomposable, and there exists some $\Ft{t}$-measurable random set $\tilde R_t(0)$ such that $A_t \cap M_t = \LdpK{p}{t}{\tilde R_t(0)}$ and $\P(\tilde R_t(0) = M) = 0$.
\end{definition}

The acceptance set of a conditional risk measure $R_t$ is given by 
$A_t = \lrcurly{X \in \LdpF{} \; | \; 0 \in R_t(X)}$. 
For any conditional acceptance set $A_t$,  the function
$R_t(X) = \lrcurly{u \in M_t \; | \; X + u \in A_t}$
is a conditional risk measure. Further this relation is one-to-one, i.e.\ we can consider a $(R_t,A_t)$ pair or equivalently just one of the two.
Given a risk measure and acceptance set pair $(R_t,A_t)$ then the following properties hold.  Many of these results are presented in Proposition 2.11 in~\cite{FR12}; those that are new here are trivial by the primal representation of the risk measure and acceptance set.
\begin{itemize}
\item $R_t$ is $\K$-compatible if and only if $A_t + \K = A_t$;
\item $R_t$ is local if and only if $A_t$ is $\Ft{t}$-decomposable, i.e.\ $\sum_{n = 1}^N 1_{\Omega_t^n} X_n \in A_t$ for any finite partition $(\Omega_t^n)_{n = 1}^N \subseteq \Ft{t}$ of $\Omega$ and any family $(X_n)_{n = 1}^N \subseteq A_t$.
\item $R_t$ is conditionally convex if and only if $A_t$ is conditionally convex;
\item $R_t$ is conditionally positive homogeneous if and only if $A_t$ is a conditional cone;
\item $R_t$ has a closed graph if and only if $A_t$ is closed;
\item $R_t$ is convex upper continuous if and only if $A_t - D := \{X-d \; | \; X \in A_t, d \in D\}$ is closed for every $D \in \mathcal{G}(M_t;M_{t,-})$;
\item $R_t$ is halfspace lower continuous if and only if $A_t + \{u \in M_t \; | \; \E{\trans{w}u} > 0\}$ is open for every $w \in M_{t,+}^+ \backslash \prp{M_t}$.  If $R_t$ is local then the equivalence only requires $w \in \LdpK{q}{t}{M_+^+ \backslash \prp{M}}$. 
\end{itemize}

Before we continue, we wish to introduce and consider stepped risk measures.  That is, $R_{t,s}: M_s \to \mathcal{P}(M_t;M_{t,+})$ for $t \leq s$ as discussed in~\cite[Section~8.3]{FR12b}.  We denote and define the stepped acceptance set by $A_{t,s} := A_t \cap M_s$.  These stepped risk measures and acceptance sets are integral to a recursive relation for time consistency as discussed in \cite{FR12} and replicated in Theorem~\ref{thm_mptc_equiv} below.

\section{Scalarized risk measures}
\label{sec:scalar}

In this section, we consider the scalarizations of set-valued risk measures.  We recall some results on the primal and dual representation of multivariate scalar risk measures from~\cite{FR13-survey,FR15-supermtg}, and deduce some further results (Proposition~\ref{prop_finite} and~\ref{prop_attained} below). Note that in Appendix~\ref{sec:dual-aux} we provide also a new, auxiliary dual representation of multivariate scalar risk measures that splits the dual variables $(\Q,m_{\perp}) \in \W_t(w)$ into a stepped part from time $t$ to $s$ and a second set of dual variables that exists at time $s$.  This auxiliary dual representation is important for deducing the main result on time consistency in Section~\ref{sec:mptc}.

We now wish to introduce the multivariate scalar risk measures.  We will define these in the same manner as done in~\cite{FR13-survey,FR15-supermtg,FR18-1asset} via the scalarization of a set-valued risk measure.  
That is, we consider the primal representation -- given in \eqref{defn_scalar} below -- as the starting point.
\begin{definition}
\label{defn_scalar}
A \textbf{\emph{multivariate conditional scalarized risk measure}} is a mapping $\rho_t^{M,w}: \LdpF{} \to \barLzF{t}$ associated with a set-valued risk measure $R_t: \LdpF{} \to \mathcal{P}(M_t;M_{t,+})$ and normal direction $w \in \LdpK{q}{t}{M_+^+ \backslash \prp{M}}$ which satisfies:
\begin{align}
\rho_t^{M,w}(X) := \essinf\lrcurly{\trans{w}u \; | \; u \in R_t(X)} = \essinf\lrcurly{\trans{w}u \; | \; u \in M_t, \; X + u \in A_t}.
\end{align}
\end{definition}
Throughout much of this text we will omit the subspace parameter $M$ from the superscript of multivariate scalarized risk measures when this choice is clear from context.  Additionally, we will sometimes call the normal direction $w$ a pricing vector as, under markets with transaction costs, it must be a consistent price with respect to the market model defined by a solvency cone $\K$.

In the following proposition we show that the multivariate conditional scalar risk measures satisfy monotonicity and a translative property.  These properties are usually given as the definition of a risk measure in the literature (see e.g.\ \cite{FR13-survey,FR18-1asset}). However, here we consider the primal representation via set-valued risk measures as the starting point.  The choice to consider set-valued risk measures will become clearer with the consideration of time consistency properties in the next section.

\begin{proposition}[Proposition 2.30 of \cite{FR13-survey}]
\label{prop_multi-scalar_properties}
Let $\rho_t^w: \LdpF{} \to \barLzF{t}$ be a multivariate conditional scalar risk measure at time $t$ for pricing vector $w \in \LdpK{q}{t}{M_+^+ \backslash \prp{M}}$.  Then $\rho_t^w$ satisfies the following conditions.
\begin{enumerate}
\item If $X,Y \in \LdpF{}$ such that $Y - X \in \LdpK{p}{}{\R^d_+}$ then $\rho_t^w(Y) \leq \rho_t^w(X)$.
\item If $X \in \LdpF{}$ and $m \in M_t$ then $\rho_t^w(X + m) = \rho_t^w(X) - \trans{w}m$.
\end{enumerate}

Further, if we consider the family of such risk measures over all pricing vectors $w \in \LdpK{q}{t}{M_+^+ \backslash \prp{M}}$ then we
have the following finiteness properties.
\begin{enumerate}
\setcounter{enumi}{2}
\item $\rho_t^w(0) < +\infty$ for every $w \in \LdpK{q}{t}{M_+^+ \backslash \prp{M}}$.
\item $\rho_t^w(0) > -\infty$ for some $w \in \LdpK{q}{t}{M_+^+ \backslash \prp{M}}$.
\end{enumerate}
\end{proposition}

We refer to \cite{FR13-survey} for further properties relating the set-valued risk measures with the family of multivariate scalarized risk measures.  Below we copy one such result demonstrating that convexity and coherence of the set-valued risk measures imply the same properties for the associated scalarized risk measures.
\begin{corollary}[Corollary~3.17 of \cite{FR13-survey}]
\label{cor_vector_scalar}
Let $R_t: \LdpF{} \to \mathcal{P}(M_t;M_{t,+})$ be a conditionally convex (conditionally positive
homogeneous) conditional risk measure at time $t$ with closed and $\Ft{t}$-decomposable images, then the associated 
multivariate scalarized risk measure $\rho_t^w: \LdpF{} \to \barLzF{t}$ is conditionally convex (resp.\ conditionally positive homogeneous).
\end{corollary}

\begin{assumption}\label{ass:convex}
For the remainder of this work we will assume that the underlying set-valued risk measure is normalized, c.u.c., and conditionally convex. 
Though we assume c.u.c., this can be weakened as we will see in some of the examples.
\end{assumption}

In prior works, see \cite{FR13-survey}, the normal direction parameter $w$ has been defined as being in the larger space $w \in M_{t,+}^+ \backslash \prp{M_t}$.  However, due to the construction of the multivariate scalarized risk measures we have that $\rho_t^w(X) = 0$ on $\{\omega \in \Omega \; | \; w(\omega) \in \prp{M}\}$.  Therefore, without loss of generality we can consider those normal directions that almost surely do not take value in $\prp{M}$, i.e.\ $w \in \LdpK{q}{t}{M_+^+ \backslash \prp{M}}$.

Note that when $R_t$ is $\K_t$-compatible for some $\Ft{t}$-decomposable closed cone with $M_{t,+} \subseteq \K_t \subsetneq M_t$, then $\rho_t^w(X) = -\infty$ on $\{\omega \in \Omega \; | \; w(\omega) \not\in \K_t[\omega]^+\}$ for every $X \in \LdpF{}$.  As such, we would be able to restrict the set of parameters to $w \in \LdpK{q}{t}{K_t^+ \backslash \prp{M}}$ where $\K_t = \LdpK{p}{t}{K_t}$; such a representation exists by \cite[Theorem 2.1.6]{M05}.  In particular, by $R_t$ normalized, we need only consider $w \in \LdpK{q}{t}{\tilde R_t(0)^+ \backslash \prp{M}}$ where $R_t(0) = \LdpK{p}{t}{\tilde R_t(0)}$ is a closed convex cone.  In fact, the properties of $R_t$ assumed in Assumption~\ref{ass:convex} imply normalization of all scalarized risk measures, i.e.\ $\rho_t^w(0) = 0$ for any $w \in \LdpK{q}{t}{\tilde R_t(0)^+ \backslash \prp{M}}$.

Before continuing, we wish to recall equivalence between the closed and conditionally convex set-valued risk measures and the family of scalarized risk measures from \cite[Lemma~3.18]{FR13-survey} (and utilizing results from \cite[Theorem 3.3]{FR13-survey} on the measurable selector approach to set-valued risk measures)
\begin{equation}
\label{eq_set-scalar}
R_t(X) = \bigcap_{w \in \LdpK{q}{t}{\tilde R_t(0)^+ \backslash \prp{M}}} \lrcurly{u \in M_t \; | \; \rho_t^w(X) \leq \trans{w}u \; \Pas} \quad \forall X \in \LdpF{}.
\end{equation}
By this representation, we can similarly define the acceptance set for the primal representation as
\[A_t = \bigcap_{w \in \LdpK{q}{t}{\tilde R_t(0)^+ \backslash \prp{M}}} \lrcurly{X \in \LdpF{} \; | \; \rho_t^w(X) \leq 0 \; \Pas}.\]

We complete our discussion of the primal representation for the multivariate scalar risk measures by considering the integrability of the images $\rho_t^w(X)$. 
\begin{proposition}
\label{prop_finite}
Let $p = +\infty$.  If $\LdiF{} = A_t + M_t$ then the multivariate scalarized risk measure is integrable, i.e.\ $\rho_t^w(X) \in \LoF{t}$ for any $w \in \LdoK{t}{\tilde R_t(0)^+ \backslash \prp{M}}$ and $X \in \LdiF{}$.  Moreover, the multivariate scalarized risk measure is integrable if the full eligible space $M = \R^d$ is taken.
\end{proposition}

\begin{proof}
The proof is given in Appendix~\ref{sec:App4.1}.
\end{proof}

\begin{assumption}\label{ass:finite}
For the remainder of this work we will assume that all scalarizations are integrable, i.e.\ $\rho_t^w(X) \in \LoF{t}$ for all times $t$, all normal directions $w \in \LdpK{q}{t}{\tilde R_t(0)^+ \backslash \prp{M}}$, and all portfolios $X \in \LdpF{}$.  In particular, by Proposition~\ref{prop_finite}, this is true if $p = +\infty$ and either $\LdiF{} = A_t + M_t$ or $M = \R^d$.
\end{assumption}


With these results we want to revisit the dual representation from \cite{FR15-supermtg} which is further studied in the special case with a single eligible asset in \cite{FR18-1asset}.  
\begin{proposition}[Proposition~A.1.7 of \cite{FR15-supermtg}]
\label{prop_scalar-cc}
Let $w \in \LdpK{q}{t}{\tilde R_t(0)^+ \backslash \prp{M}}$, then for every $X \in \LdpF{}$
\begin{equation*}
\rho_t^w(X) 
= \esssup_{(\Q,m_{\perp}) \in \W_t(w)} \lrparen{-\alpha_t(\Q,w + m_{\perp}) + \transp{w + m_{\perp}}\EQt{-X}{t}}
\end{equation*}
where 
\begin{align*}
\alpha_t(\Q,w) &:= \esssup_{Z \in A_t} \trans{w}\EQt{-Z}{t} \quad \text{and}\\
\W_t(w) &:= \lrcurly{(\Q,m_{\perp}) \in \mathcal{M}^d \times \prp{M_t} \; | \; (\Q,w + m_{\perp}) \in \W_t}.
\end{align*}
\end{proposition}

\begin{remark}
\label{rem_fullM}
If $M = \R^d$ then $\prp{M} = \{0\}$.  The dual representation thus can be provided by a single dual variable $\Q$, i.e.\ for any $X \in \LdpF{}$ and $w \in \LdpK{q}{t}{\tilde R_t(0)^+ \backslash \{0\}}$
\[\rho_t^w(X) = \esssup_{\Q \in \W_t^d(w)} \lrparen{-\alpha_t(\Q,w) + \trans{w}\EQt{-X}{t}}\]
with $\W_t^d(w) = \lrcurly{\Q \in \mathcal{M}^d \; | \; (\Q,w) \in \W_t}$.
\end{remark}

\begin{corollary}
\label{cor_scalar-cc}
Consider the setting of Proposition~\ref{prop_scalar-cc} such that $R_t$ is additionally conditionally coherent.  Then for every $X \in \LdpF{}$
\[\rho_t^w(X) = \esssup_{(\Q,m_{\perp}) \in \W_t^R(w)} \transp{w + m_{\perp}}\EQt{-X}{t}\]
where $\W_t^R(w) := \lrcurly{(\Q,m_{\perp}) \in \W_t(w) \; | \; w_t^T(\Q,w+m_{\perp}) \in A_t^+}$.
\end{corollary}
\begin{proof}
This follows identically to Corollary~2.5 of \cite{FP06} noting that $\alpha_t(\Q,w+m_{\perp}) = 0$ if and only if $w_t^T(\Q,w+m_{\perp}) \in A_t^+$.
\end{proof}

We now provide a dual representation for the stepped scalarized risk measures. 

\begin{corollary}
\label{cor_scalar-stepped}
Let $w \in \LdpK{q}{t}{\tilde R_t(0)^+ \backslash \prp{M}}$, then for every $X \in M_s$
\begin{align*}
\rho_{t,s}^w(X) &:= \essinf_{u \in R_{t,s}(X)} \trans{w}u \\
&= \esssup_{(\Q,m_{\perp}) \in \W_{t,s}(w)} \lrparen{-\alpha_{t,s}(\Q,w + m_{\perp}) + \transp{w + m_{\perp}}\EQt{-X}{t}}
\end{align*}
where 
\begin{align*}
\alpha_{t,s}(\Q,w) &:= \esssup_{Z \in A_{t,s}} \trans{w}\EQt{-Z}{t} \quad \text{and}\\
\W_{t,s}(w) &:= \lrcurly{(\Q,m_{\perp}) \in \mathcal{M}^d \times \prp{M_t} \; | \; (\Q,w + m_{\perp}) \in \W_{t,s}}.
\end{align*}

If $R_t$ is additionally conditionally coherent, then for every $X \in M_s$
\begin{equation*}
\rho_{t,s}^w(X) = \esssup_{(\Q,m_{\perp}) \in \W_{t,s}^R(w)} \transp{w + m_{\perp}}\EQt{-X}{t}
\end{equation*}
where $\W_{t,s}^R(w) := \lrcurly{(\Q,m_{\perp}) \in \W_{t,s}(w) \; | \; w_t^s(\Q,w+m_{\perp}) \in A_{t,s}^+}$.
\end{corollary}
\begin{proof}
The convex case is provided by \cite[Corollary~A.1.8]{FR15-supermtg}.  The coherent case follows identically from the logic of Corollary~\ref{cor_scalar-cc}.
\end{proof}

We now consider the case in which the underlying set-valued risk measure satisfies a stronger continuity property.  This is undertaken so as to find the dual representation is attained by some choice of dual variables $(\Q,m_{\perp})$. 

\begin{proposition}
\label{prop_attained}
Consider the setting of Proposition~\ref{prop_scalar-cc} such that $R_t$ is additionally h.l.c. 
Let $w \in \LdpK{q}{t}{\tilde R_t(0)^+ \backslash \prp{M}}$, then for every $X \in \LdpF{}$
\begin{equation*}
\rho_t^w(X) = \essmax_{(\Q,m_{\perp}) \in \W_t(w)} \lrparen{-\alpha_t(\Q,w+m_{\perp}) + \transp{w+m_{\perp}}\EQt{-X}{t}}
\end{equation*}
and for any time $s \in (t,T]$ and every $X \in M_s$
\begin{equation*}
\rho_{t,s}^w(X) = \essmax_{(\Q,m_{\perp}) \in \W_{t,s}(w)} \lrparen{-\alpha_{t,s}(\Q,w + m_{\perp}) + \transp{w + m_{\perp}}\EQt{-X}{t}}.
\end{equation*}
If $R_t$ is additionally conditionally coherent, then for every $X \in \LdpF{}$
\begin{equation*}
\rho_t^w(X) = \essmax_{(\Q,m_{\perp}) \in \W_t^R(w)} \transp{w + m_{\perp}}\EQt{-X}{t}
\end{equation*}
and for any time $s \in (t,T]$ and every $X \in M_s$
\begin{equation*}
\rho_{t,s}^w(X) = \essmax_{(\Q,m_{\perp}) \in \W_{t,s}^R(w)} \transp{w + m_{\perp}}\EQt{-X}{t}.
\end{equation*}
\end{proposition}

\begin{proof}
The proof is given in Appendix~\ref{sec:App4.1}.
\end{proof}

\section{Multiportfolio time consistency}
\label{sec:mptc}

Given the results of the previous section on the representation of the multivariate scalarized risk measures, we now wish to consider how these risk measures relate through time.  To do this we consider multiportfolio time consistency, which was defined for set-valued risk measures in \cite{FR12,FR12b}.  The main result of this paper is in determining an equivalent time consistency property for the scalarizations to multiportfolio time consistency for the set-valued risk measure. The formal definition of this is provided below.  For our purposes, we are interested in this property due to the additive property for acceptance sets (Theorem~\ref{thm_mptc_equiv}\eqref{thm_equiv_acceptance}) as this gives a direct comparison to the scalar notion.  The main result of this section is Theorem~\ref{thm_scalar_recursive} in which a recursive formulation is provided for the scalarized risk measures.

\begin{definition}[Definition~2.7 of \cite{FR12b}]
\label{defn_mptc}
A dynamic risk measure $\seq{R}$ is called \textbf{\emph{multiportfolio time consistent}} if for all times $0 \leq t < s \leq T$, all portfolios $X\in \LdpF{}$ and all sets ${\bf Y}\subseteq \LdpF{}$ the implication
\begin{equation}
  R_s(X) \subseteq \bigcup_{Y \in {\bf Y}} R_s(Y) \Rightarrow R_t(X) \subseteq \bigcup_{Y \in {\bf Y}} R_t(Y)
\end{equation}
is satisfied.
\end{definition}
Multiportfolio time consistency means that if at some time any risk compensation portfolio for $X$ also compensates the risk of some portfolio $Y$ in the set ${\bf Y}$, then at any prior time the same relation should hold true.

The following theorem gives equivalent characterizations of multiportfolio time consistency: a recursion in the spirit of Bellman's principle (Property~\eqref{thm_equiv_recursive} below) and an additive property for the acceptance sets (Property~\eqref{thm_equiv_acceptance}).
These properties are important for the construction of multiportfolio time consistent risk measures.
\begin{theorem}[Theorem~3.4 of \cite{FR12}]
\label{thm_mptc_equiv}
For a normalized dynamic risk measure $\seq{R}$ the following are equivalent:
\begin{enumerate}
\item \label{thm_equiv_mptc} $\seq{R}$ is multiportfolio time consistent.
\item \label{thm_equiv_recursive} $R_t$ is recursive, that is for every time $0 \leq t < s \leq T$
    \begin{equation}
    \label{recursive}
        R_t(X) = \bigcup_{Z \in R_s(X)} R_t(-Z) =: R_t(-R_s(X)).
    \end{equation}
\item \label{thm_equiv_acceptance} for every time $0 \leq t < s \leq T$
    \begin{equation}
    \label{sum_acceptance}
        A_t = A_{t,s} + A_s.
    \end{equation}
\end{enumerate}
\end{theorem}

We now come to the main result of this paper.  Namely we wish to relate multiportfolio time consistency to a property for the family of scalarizations.  Notably we find a recursive formulation that relies on the dual representation of the scalarizations.  However, this recursion does not follow a single scalarization with directions $w \in M_+^+ \backslash \prp{M}$ throughout all time, but requires a moving scalarization.  This notion of moving scalarizations for time consistency aligns with recent results for dynamic multi-objective programs \cite{KMZ17,KR18}. In many cases a specific scalarization $w \in R_0(0)^+ \backslash \prp{M}$ at time $0$ is desired, for instance $w = \vec{1}$ in the case of systemic risk.  The below theorem shows that even with this specific scalar multivariate risk measure at time $0$, the entire family of scalarizations (in the future) is needed to obtain a recursive formulation derived from the usual time consistency property of the acceptance sets, that is, $A_t = A_{t,s} + A_s$, i.e., multiportfolio time consistent.
\begin{theorem}
\label{thm_scalar_recursive}
$\seq{R}$ is multiportfolio time consistent if and only if for any times $0 \leq t < s \leq T$, any $w \in \LdpK{q}{t}{\tilde R_t(0)^+ \backslash \prp{M}}$, and any $X \in \LdpF{}$
\begin{equation}
\label{eq_scalar_recursive}
\rho_t^w(X) = \esssup_{\substack{(\Q,m_{\perp}) \in \W_{t,s}(w)\\ w_t^s(\Q,w+m_{\perp}) \in R_s(0)^+}} \lrparen{-\alpha_{t,s}(\Q,w+m_{\perp}) +
\Et{\rho_s^{w_t^s(\Q,w+m_{\perp})}(X)}{t}}.
\end{equation}
If $\seq{R}$ is additionally conditionally coherent then for any times $0 \leq t < s \leq T$, any $w \in \LdpK{q}{t}{\tilde R_t(0)^+ \backslash \prp{M}}$, and any $X \in \LdpF{}$
\[\rho_t^w(X) = \esssup_{\substack{(\Q,m_{\perp}) \in \W_{t,s}(w)\\ w_t^s(\Q,w+m_{\perp}) \in R_s(0)^+ \cap A_{t,s}^+}} \Et{\rho_s^{w_t^s(\Q,w+m_{\perp})}(X)}{t}\]
where $\rho_s^{w_s}(X) \equiv -\infty$ if $\W_s^R(w_s) = \emptyset$. 
\end{theorem}
\begin{proof}
The convex setting is an immediate consequence of Theorem~\ref{thm_mptc_equiv} and Lemma~\ref{lemma_scalar}.  If $\seq{R}$ is additionally conditionally coherent then the result is a direct consequence of $\beta_{t,s}^w(\Q,m_{\perp},\R,n_{\perp}) = 0$ on $\{w_s^T(\R,w_t^s(\Q,w+m_{\perp})+n_{\perp}) \in A_t^+\}$ and is $-\infty$ otherwise.
\end{proof}
Note that the recursive formulation~\eqref{eq_scalar_recursive} takes the $\P$-conditional expectation of the future risk measures $\Et{\rho_s^{w_t^s(\Q,w+m_{\perp})}(X)}{t}$ rather than the $\Q$-conditional expectation as it appears in the dual representation and as would be expected in the univariate setting.  However, this recursive formulation can equivalently be given under the $\Q$-conditional expectation with the modification
\[\rho_t^w(X) = \esssup_{\substack{(\Q,m_{\perp}) \in \W_{t,s}(w)\\ w_t^s(\Q,w+m_{\perp}) \in R_s(0)^+}} \lrparen{\begin{array}{l}-\alpha_{t,s}(\Q,w+m_{\perp})\\ \quad + \transp{w+m_{\perp}}\EQt{w_t^s(\Q,w+m_{\perp})^{\ddag} \rho_s^{w_t^s(\Q,w+m_{\perp})}(X)}{t}\end{array}}\]
where $v^{\ddag} := v(\trans{v}v)^{-1}$ on $\{\trans{v}v > 0\}$ and $0$ on $\{\trans{v}v = 0\}$ is an extension of the usual pseudoinverse.

\begin{remark}
\label{rem_fullM_recursive}
As in Remark~\ref{rem_fullM}, if $M = \R^d$ then $\prp{M} = \{0\}$.  The recursive formulation thus can be simplified as it no longer needs to consider the second dual variable $m_{\perp} \in \prp{M_t}$.  In fact, in the setting with $M = \R^d$, multiportfolio time consistency is equivalent to the recursive formulation: for any times $0 \leq t < s \leq T$, any $w \in \LdpK{q}{t}{\tilde R_t(0)^+ \backslash \prp{M}}$, and any $X \in \LdpF{}$
\begin{equation*}
\rho_t^w(X) = \esssup_{\substack{\Q \in \W_{t,s}^d(w)\\ w_t^s(\Q,w) \in R_s(0)^+}} \lrparen{-\alpha_{t,s}(\Q,w) +
\Et{\rho_s^{w_t^s(\Q,w)}(X)}{t}}.
\end{equation*}
\end{remark}

\begin{remark}
\label{rem:moving_scalar}
Consider a discrete time setting.
When the essential supremum in the recursive formulation is attained (e.g., under h.l.c.\ of the underlying set-valued risk measure $R_t$ as shown in Proposition~\ref{prop_attained} or if the probability space if finite) then the sequence of moving scalarizations can be made explicit through the maximizing dual variables over time.  
In particular, this maximizing sequence exactly provides the moving scalarizations given a portfolio $X \in \LdpF{}$ and an initial (time 0) normal direction $w \in R_0(0)^+ \backslash \prp{M}$. The moving scalarizations are then provided by the sequence $\seq{w}$ where $w_0 = w$ and $w_{t+1} = w_t^{t+1}(\Q^t , w_t + m_{\perp}^t)$ and $(\Q^t,m_{\perp}^t)$ are the maximizers of
\[(\Q^t,m_{\perp}^t) \in \argessmax_{\substack{(\Q,m_{\perp}) \in \W_{t,t+1}(w_t)\\ w_t^{t+1}(\Q,w_t+m_{\perp}) \in R_{t+1}(0)^+}} \lrparen{-\alpha_{t,t+1}(\Q,w_t+m_{\perp}) + \Et{\rho_{t+1}^{w_t^{t+1}(\Q,w_t+m_{\perp})}(X)}{t}}.\]
However, in the general case, though we find the notion of moving scalarizations, there may not exist an actual sequence of scalarizations that are consistent in time as the (essential) supremum might not be attained.
\end{remark}

\begin{remark}
\label{rem_m=1}
Consider the special case that $M = \R \times \{0\}^{d-1}$ and $p = +\infty$ from \cite{FR18-1asset}.  In that paper it was shown that multiportfolio time consistency (under this single asset eligible space) was equivalent to
\[\rho_s^{e_1}(X) \leq \rho_s^{e_1}(Y) \; \Rightarrow \; \rho_t^{e_1}(X) \leq \rho_t^{e_1}(Y) \quad \forall t < s\]
where $e_1 = \transp{1,0,...,0} \in \R^d$.  
In fact, by construction we can find that $\rho_t^w(X) = w_1 \rho_t^{e_1}(X)$ for any $w \in \LdpK{q}{t}{M_+^+ \backslash \prp{M}}$ (further noting that $R_t(0) = M_{t,+}$ in this setting by Assumption~\ref{ass:convex}).
Additionally, since $e_1 \in \LdiF{t}$ for all times $t$, we can modify the proof of Proposition~\ref{prop_finite} in order to show that $\rho_t^{e_1}(X) \in \LiF{t}$ for all times $t$.
Therefore the recursive formulation of Theorem~\ref{thm_scalar_recursive} simplifies to
\begin{align*}
\rho_t^{e_1}(X) &= \esssup_{(\Q,m_{\perp}) \in \W_{t,s}(e_1)} \lrparen{-\alpha_{t,s}(\Q,e_1+m_{\perp}) + \EQnt{1}{\rho_s^{e_1}(X)}{t}}\\
&= \esssup_{(\Q,m_{\perp}) \in \W_{t,s}(e_1)} \lrparen{-\alpha_{t,s}(\Q,e_1+m_{\perp}) + \transp{e_1+m_{\perp}}\EQt{\rho_s^{e_1}(X)e_1}{t}}\\
&= \rho_{t,s}^{e_1}(-\rho_s^{e_1}(X)e_1)
\end{align*}
and no other scalarizations need to be considered.
While this definition and recursive formulation match the usual ones in the literature for time consistency for multivariate risk measures (see \cite{HMBS16}), we wish to note that the usual examples, such as the superhedging risk measure, do not satisfy multiportfolio time consistency in this setting (i.e.\ for $M = \R \times \{0\}^{d-1}$) with $d > 1$.
\end{remark}

We conclude this section by presenting the formal results relating the multiportfolio time consistency and the recursive structure for scalarizations. In doing so we also deduce the cocycle condition on the penalty functions as in the $d = 1$ case of \cite{FP06,CDK06,BN08,BN09}.  These below results provide us with the foundation for the proof of Theorem~\ref{thm_scalar_recursive}.
\begin{lemma}
\label{lemma_scalar}
Fix times $0 \leq t < s \leq T$.  The following are equivalent:
\begin{enumerate}
\item\label{thm_scalar-acceptance} $A_t \subseteq A_{t,s} + A_s$ (resp.\ $\supseteq$);
\item\label{thm_scalar-recursive} For every $X \in \LdpF{}$ and $w \in \LdpK{q}{t}{\tilde R_t(0)^+ \backslash \prp{M}}$
\[\rho_t^w(X) \geq \esssup_{\substack{(\Q,m_{\perp}) \in \W_{t,s}(w)\\ w_t^s(\Q,w+m_{\perp}) \in R_s(0)^+}} \lrparen{\begin{array}{l} -\alpha_{t,s}(\Q,w+m_{\perp})\\ \quad + \Et{\rho_s^{w_t^s(\Q,w+m_{\perp})}(X)}{t} \end{array}} \quad \text{(resp.\ $\leq$)};\] 
\item\label{thm_scalar-penalty} For every $w \in \LdpK{q}{t}{\tilde R_t(0)^+ \backslash \prp{M}}$, $(\Q,m_{\perp}) \in \W_{t,s}(w)$, and $(\R,n_{\perp}) \in \W_s(w_t^s(\Q,w+m_{\perp}))$
\[\begin{array}{ll} \beta_{t,s}^w(\Q,m_{\perp},\R,n_{\perp}) &\leq \alpha_{t,s}(\Q,w+m_{\perp})\\ &\qquad + \Et{\alpha_s(\R,w_t^s(\Q,w+m_{\perp})+n_{\perp})}{t} \end{array} \quad \text{(resp.\ $\geq$)},\]
where
\[\beta_{t,s}^w(\Q,m_{\perp},\R,n_{\perp}) = \esssup_{Z \in A_t} \Et{\trans{w_s^T(\R,w_t^s(\Q,w+m_{\perp})+n_{\perp})}(-Z)}{t}\]
for any $(\Q,m_{\perp}) \in \W_t(w)$ and $(\R,n_{\perp}) \in \W_s(w_t^s(\Q,w+m_{\perp}))$, which arises in the auxiliary dual representation utilizing stepped dual variables as given in Proposition~\ref{prop_extended_dual}.
\end{enumerate}
\end{lemma}
\begin{proof}
The proof is given in Appendix~\ref{sec:App5}.
\end{proof}

\section{Examples}
\label{sec:examples}
Many examples of multiportfolio time consistent set-valued risk measures have been presented in the literature.  See \cite{FR12,FR12b,FR14-alg} for some specific risk measures.  In this section we will consider two examples in detail.  First, we will introduce the superhedging risk measure with proportional and convex transaction costs to demonstrate the recursive formulation provided in Theorem~\ref{thm_scalar_recursive}.  Second, we will consider a composed risk measure (as done in \cite{CK10} for univariate risk measures), with an emphasis on the composed average value at risk.
For these examples we will restrict ourselves to the full eligible space $M = \R^d$ and with $p = +\infty$ for simplicity.  In particular, this allows us to drop the second dual variable $m_{\perp}$ as it must be 0 as discussed in Remark~\ref{rem_fullM}.

\subsection{Superhedging risk measure}
\label{sec:shp}
Consider a market in discrete time with convex transaction costs, e.g.\ with a bid-ask spread and market impacts from trading.  We will model this market via a sequence of convex solvency regions $\seq{K}$ for $\Ft{t}$-measurable closed and convex random sets $K_t$ where $\R^d_+ \subseteq K_t[\omega] \subsetneq \R^d$ for almost any $\omega \in \Omega$ and time $t \in \{0,1,...,T\}$.  The solvency regions denote those portfolios that can be traded in the market into the 0 portfolio (with the ability to ``throw away'' assets).  If $K_t$ is additionally conical then this market is one with proportional transaction costs only.

The superhedging risk measure, as described in \cite{FR12,FR14-alg,LR11}, is provided by the primal representation
\[SHP_t(X) = \lrcurly{u \in \LdiF{t} \; | \; X + u \in \sum_{s = t}^T \LdiK{s}{K_s}}\]
for any $X \in \LdiF{}$ and any time $t$.

Under a suitable no-arbitrage argument (no-scalable robust no-arbitrage \cite{PP10}) the superhedging risk measure is closed.  Though the superhedging risk measure may not be c.u.c., all results of this work still hold due to the specific structure of this acceptance set and a no-arbitrage argument.  Thus we are able to determine the primal and dual representations for the scalarizations of the superhedging risk measure at time $t$ given a normal direction $w \in \LdoK{t}{\R^d_+ \backslash \{0\}}$:
\begin{align*}
shp_t^w(X) &= \essinf\lrcurly{\trans{w}u \; | \; u \in \LdiF{t}, \; X + u \in \sum_{s = t}^T \LdiK{s}{K_s}}\\
&= \esssup_{\Q \in \W_t^d(w)} \lrparen{\sum_{s = t}^T \sigma_t^s(w_t^s(\Q,w)) + \trans{w}\EQt{-X}{t}},\\
\sigma_t^s(w_s) &= \essinf_{k \in \LdiK{s}{K_s}} \Et{\trans{w_s}k}{t}.
\end{align*}
If the market has proportional transaction costs only then $\sigma_t^s(w_s) = 0$ on $\{w_s \in K_s^+\}$ and $-\infty$ otherwise.

Finally, as it is well known that the superhedging risk measure is multiportfolio time consistent \cite{FR12,FR12b}, we are able to deduce a recursive formulation for the scalarizations:
\begin{align*}
shp_t^w(X) &= \esssup_{\Q \in \W_t^d(w)} \lrparen{\sigma_t^t(w) + \Et{shp_{t+1}^{w_t^{t+1}(\Q,w)}(X)}{t}}\\
&= \sigma_t^t(w) + \esssup_{\Q \in \W_t^d(w)} \Et{shp_{t+1}^{w_t^{t+1}(\Q,w)}(X)}{t}
\end{align*}
for any $X \in \LdiF{}$ and $w \in \LdoK{t}{\R^d_+ \backslash \{0\}}$.  In the proportional transaction costs case we recover $shp_t^w(X) = \esssup_{\Q \in \W_t^d(w)} \Et{shp_{t+1}^{w_t^{t+1}(\Q,w)}(X)}{t}$ if $w \in \LdoK{t}{K_t^+}$ and $-\infty$ otherwise.

\subsection{Composed risk measures}
\label{sec:composed}
Consider now the backwards composition of set-valued risk measures.  Such a composition, defined in \cite{FR12}, guarantees that the resultant risk measure is multiportfolio time consistent.  We define this backwards composition $\seq{\hat R}$ for a sequence of one-step risk measure $(R_{t,t+1})_{t = 0}^{T-1}$ and a terminal risk measure $R_T$ as:
\begin{align*}
\hat R_t(X) &:= R_{t,t+1}(\hat R_{t+1}(X)) = \bigcup_{Z \in \hat R_{t+1}(X)} R_{t,t+1}(-Z) \quad \forall t \in \{0,1,...,T-1\}\\
\hat R_T(X) &= R_T(X).
\end{align*}
As considered below, this backwards composition will be applied to the average value at risk (see also \cite{FR12,FR12b,FR14-alg}). This formulation could be applied, as well, to systemic risk measures in order to construct a multiportfolio time consistent risk measure as the dynamic version of set-valued systemic risk measures have only recently been introduced~\cite{doldi2020conditional}.

For any $w \in \LdpK{q}{}{\tilde R_T(0) \backslash \{0\}}$, the scalarization at the terminal time is given by $\hat \rho_T^w(X) = -\trans{w}X$.  At any prior time $t \in \{0,1,...,T-1\}$, the scalarizations of such a composition with respect to $w$ is given by
\begin{align*}
\hat \rho_t^w(X) &= \essinf\lrcurly{\trans{w}u \; | \; u \in \hat R_t(X)}\\
&= \esssup_{\Q \in \W_t^d(w)} \lrparen{-\sum_{s = t}^{T-1} \alpha_{s,s+1}(\Q,w_t^s(\Q,w)) - \alpha_T(\Q,w_t^T(\Q,w)) + \trans{w}\EQt{-X}{t}}\\
&= \esssup_{\Q \in \W_t^d(w)} \lrparen{-\alpha_{t,t+1}(\Q,w) + \Et{\hat \rho_{t+1}^{w_t^{t+1}(\Q,w)}(X)}{t}}.
\end{align*}
where the stepped penalty functions $\alpha_{s,s+1}$ are defined with respect to the original stepped acceptance sets $A_{s,s+1}$.
Recall that $\W_t^d$ is defined in Remark~\ref{rem_fullM}.  Thus we see that the composition of the scalarizations backwards in time requires knowledge of the entire family of scalarizations.

We will now consider the composed average value at risk to demonstrate the representations and recursive structure for a specific backwards composition.  
This risk measure has been studied in detail in \cite[Section 6.1]{FR12b}.
Briefly, fix some lower threshold $\epsilon > 0$; we will utilize the stepped average value at risk with levels $\lambda^t \in \LdiK{t}{[\epsilon,1]^d}$ for the stepped risk measure from time $t$ to $t+1$.  
The scalarization $\hat \rho_t^w$ of the composed average value at risk at time $t$ with normal direction $w$ has the dual representation:
\begin{align*}
\hat \rho_t^w(X) &= \esssup_{\Q \in \W_t^{AV@R}(w)} -\trans{w}\EQt{X}{t}\\
\W_t^{AV@R}(w) &= \lrcurly{\Q \in \mathcal{M}^d \; | \; 1_{\{w_i > 0\}} \bar\xi_{s,s+1}(\Q_i) \leq \frac{1}{\lambda_i^s} \; \forall i \in \{1,...,d\},\; s \in \{t,...,T-1\}}.
\end{align*}
By multiportfolio time consistency and Theorem~\ref{thm_scalar_recursive}, we recover the recursive representation:
\begin{align*}
\hat \rho_t^w(X) &= \esssup_{\Q \in \W_{t,t+1}^{AV@R}(w)} \Et{\hat \rho_{t+1}^{w_t^{t+1}(\Q,w)}(X)}{t}\\
\W_{t,t+1}^{AV@R}(w) &= \lrcurly{\Q \in \mathcal{M}^d \; | \; 1_{\{w_i > 0\}} \bar\xi_{t,t+1}(\Q_i) \leq \frac{1}{\lambda_i^t} \; \forall i = \{1,...,d\}}
\end{align*}
for any time $t = 0,1,...,T-1$.  See \cite[Section 2.3.1]{CK10} to compare the set of dual variables $\W_{t,t+1}^{AV@R}(w)$ and $\W_t^{AV@R}(w)$ to the univariate setting.

\section{Conclusion}
\label{sec:conclusion}
In this paper we consider the scalarizations of dynamic set-valued risk measures.  We first summarize results on the dual representations of these multivariate scalar risk measures that have been used previously in the literature.  We then introduce a new dual representation which we utilize to consider the time consistency of these scalarized risk measures.  The main results of this work are in determining a recursive relationship for the family of scalarizations that is equivalent to multiportfolio time consistency.  In particular, we find that this recursive structure explicitly demonstrates the moving scalarization inherent in the time consistency of many multi-objective and vector dynamic programming problems \cite{KMZ17,KR18}.  From this work and the aforementioned papers, the notion of moving scalarizations appears to be a fundamental one in Bellman's principle for multi-objective problems.  We conclude this work by considering examples to demonstrate the recursive formulation.

\appendix

\section{Proofs and auxiliary results for Section~\ref{sec:scalar}}
\subsection{Proofs of Section~\ref{sec:scalar}}
\label{sec:App4.1}
\begin{proof}[Proof of Proposition~\ref{prop_finite}.]
First we will consider the setting with a full space of eligible assets $M = \R^d$.
Fix $w \in \LdoK{t}{\tilde R_t(0)^+ \backslash \{0\}}$ and $X \in \LdiF{}$.  Define $\vec{1} := \transp{1,...,1} \in \R^d$, then by monotonicity, translativity, and normalization
\begin{align*}
\rho_t^w(X) &\leq \rho_t^w(-\|X\|_{\infty} \vec{1}) = \rho_t^w(0) + \|X\|_{\infty} \sum_{i = 1}^d w_i = \|X\|_{\infty} \sum_{i = 1}^d w_i \in \LoF{t}\\
\rho_t^w(X) &\geq \rho_t^w(\|X\|_{\infty} \vec{1}) = \rho_t^w(0) - \|X\|_{\infty} \sum_{i = 1}^d w_i = -\|X\|_{\infty} \sum_{i = 1}^d w_i \in \LoF{t}.
\end{align*}

Now consider the general eligible space $M$.  Fix $w \in \LdoK{t}{\tilde R_t(0)^+ \backslash \prp{M}}$.  
First, by \cite[Proposition 2.11(viii)]{FR12} the assumptions of this proposition imply that $R_t(X) \neq \emptyset$ for every $X \in \LdiF{}$.  Thus, by definition for any $X \in \LdiF{}$ and fixing any $u^*(X) \in R_t(X) \subseteq M_t$
\[\rho_t^{M,w}(X) = \essinf\lrcurly{\trans{w}u \; | \; u \in R_t(X)} \leq \trans{w}u^*(X) \in \LoF{t}.\]
Second, define the set-valued risk measure $\bar R_t: \LdiF{} \to \mathcal{P}(\LdiF{t};\LdiK{t}{\R^d_+})$ by 
\[\bar R_t(X) := \lrcurly{u \in \LdiK{t}{\R^m \times \R^{d-m}_+} \; | \; \transp{u_1,...,u_m,0,...,0} \in R_t(X)}.\]
By construction $\bar R_t$ is normalized, closed, and conditionally convex.  Additionally, the normal direction $w$ can be decomposed as $w = v + m_{\perp}$ for $v := \transp{w_1,...,w_m,0,...,0} \in \LdoK{t}{\tilde{\bar R}_t(0)^+ \backslash \{0\}}$ and $m_{\perp} := \transp{0,...,0,w_{m+1},...,w_d} \in \prp{M_t}$.
Thus, defining $\bar \rho_t^{\R^d,v}$ as the scalarization of $\bar R_t$ with full eligible space, we find
\[\rho_t^{M,w}(X) = \rho_t^{M,v}(X) \geq \bar \rho_t^{\R^d,v}(X) \in \LoF{t}.\]
\end{proof}

\begin{proof}[Proof of Proposition~\ref{prop_attained}.]
We will prove this result for the conditionally convex case for $\rho_t^w$ only.  The conditionally coherent and stepped cases follow identically with the representations given in Corollaries~\ref{cor_scalar-cc} and~\ref{cor_scalar-stepped}.

Define $\bar \rho_t^w: \LdpF{} \to \R$ by $\bar \rho_t^w(X) := \E{\rho_t^w(X)} = \inf_{u \in R_t(X)} \E{\trans{w}u}$ for any normal direction $w \in \LdpK{q}{t}{\tilde R_t(0)^+ \backslash \prp{M}}$ and any portfolio $X \in \LdpF{}$.  By \cite[Proposition A.1.3]{FR15-supermtg}, c.u.c.\ of the underlying set-valued risk measure $R_t$ implies $\bar \rho_t^w$ is lower semicontinuous.  We now wish to prove that $\bar \rho_t^w$ is upper semicontinuous as well.  To do so we will minimally modify the proof of \cite[Proposition 3.27]{HS12-uc} so as to only require h.l.c.\ and not full lower continuity.
For any $\epsilon > 0$, there exists some $u_{\epsilon} \in R_t(X)$ such that 
\[\E{\trans{w}u_{\epsilon}} < \inf_{u \in R_t(X)} \E{\trans{w}u} + \epsilon = \bar \rho_t^w(X) + \epsilon.\]  
This implies $V := \lrcurly{u \in M_t \; | \; \E{\trans{w}u} < \bar \rho_t^w(X) + \epsilon}$ is an open neighborhood of $u_{\epsilon}$.  Therefore, by h.l.c., $R_t^-[V]$ is an open neighborhood of $X$ and, by definition, $\bar \rho_t^w(Y) < \bar \rho_t^w(X) + \epsilon$ for any $Y \in R_t^-[V]$, i.e.\ upper semicontinuity.  

Therefore, by an application of Theorems 2.4.2(v) and 2.4.12 of \cite{Z02} to the dual representation given by \cite[Proposition A.1.1]{FR15-supermtg}, we find that
\[\bar \rho_t^w(X) = \max_{(\Q,m_{\perp}) \in \W_t(w)} \inf_{Z \in A_t} \E{\transp{w + m_{\perp}}\EQt{Z-X}{t}}\]
for any $w \in \LdpK{q}{t}{\tilde R_t(0)^+ \backslash \prp{M}}$ and $X \in \LdpF{}$.  Fixing some normal direction $w$ and portfolio $X$, let $(\Q^*,m_{\perp}^*)$ be maximizing arguments for the dual representation of $\bar \rho_t^w(X)$.  By the dual representation of $\rho_t^w(X)$, it follows that 
\[\rho_t^w(X) \geq -\alpha_t(\Q^*,w+m_{\perp}^*) + \transp{w+m_{\perp}^*}\EPt{\Q^*}{-X}{t}.\]
Finally, by construction of $\bar \rho_t^w(X) = \E{\rho_t^w(X)}$ and the $\Ft{t}$-decomposability of the acceptance set $A_t$ it follows that
\[\E{\rho_t^w(X)} = \E{-\alpha_t(\Q^*,w+m_{\perp}^*) + \transp{w+m_{\perp}^*}\EPt{\Q^*}{-X}{t}},\]
which implies that $(\Q^*,m_{\perp}^*)$ must also be maximizing arguments for the dual representation provided in Proposition~\ref{prop_scalar-cc}.
\end{proof}

\subsection{An auxiliary dual representation for scalarized risk measures}
\label{sec:dual-aux}

In this appendix we provide an auxiliary dual representation that splits the dual variables $(\Q,m_{\perp}) \in \W_t(w)$ into a stepped part from time $t$ to $s$ and a second set of dual variables that exists at time $s$.  This dual representation is used extensively in providing a time consistency relation in the Section~\ref{sec:mptc}. 
\begin{proposition}
\label{prop_extended_dual}
Let $0 \leq t \leq s \leq T$.  
Fix $w \in \LdpK{q}{t}{M_+^+ \backslash \prp{M}}$.  Then for any $X \in \LdpF{}$ it follows that
\begin{align}
\label{eq_extended_dual} \rho_t^w(X) &= \esssup_{(\Q,m_{\perp}) \in \W_{t,s}(w)} \esssup_{(\R,n_{\perp}) \in \W_s(w_t^s(\Q,w+m_{\perp}))} \lparen{-\beta_{t,s}^w(\Q,m_{\perp},\R,n_{\perp})}\\
\nonumber &\qquad \rparen{+ \Et{\trans{w_s^T(\R,w_t^s(\Q,w+m_{\perp})+n_{\perp})}(-X)}{t}}\\
\label{eq_extended_dual-2} &= \esssup_{(\Q,m_{\perp},\R,n_{\perp}) \in \tildeW_{t,s}(w)} \lparen{-\beta_{t,s}^w(\Q,m_{\perp},\R,n_{\perp})}\\
\nonumber &\qquad \rparen{+ \Et{\trans{w_s^T(\R,w_t^s(\Q,w+m_{\perp})+n_{\perp})}(-X)}{t}}
\end{align}
where
\[\beta_{t,s}^w(\Q,m_{\perp},\R,n_{\perp}) = \esssup_{Z \in A_t} \Et{\trans{w_s^T(\R,w_t^s(\Q,w+m_{\perp})+n_{\perp})}(-Z)}{t}\]
for any $(\Q,m_{\perp}) \in \W_t(w)$ and $(\R,n_{\perp}) \in \W_s(w_t^s(\Q,w+m_{\perp}))$,
and where 
\[\tildeW_{t,s}(w) = \lrcurly{(\Q,m_{\perp},\R,n_{\perp}) \; \left| \; \begin{array}{l} (\Q,m_{\perp}) \in \W_{t,s}(w),\\ (\R,n_{\perp}) \in \W_s(w_t^s(\Q,w+m_{\perp})),\\ \E{\beta_{t,s}^w(\Q,m_{\perp},\R,n_{\perp})} < +\infty \end{array}\right.}.\]
\end{proposition}
\begin{proof}
First we will show \eqref{eq_extended_dual}.
To do this, first we will show $\leq$. This trivially follows from Proposition~\ref{prop_scalar-cc} by
\begin{align*}
\rho_t^w(X) &= \esssup_{(\Q,m_{\perp}) \in \W_t(w)} \lrparen{-\alpha_t(\Q,w+m_{\perp}) + \Et{\trans{w_t^T(\Q,w+m_{\perp})}(-X)}{t}}\\
&= \esssup_{(\Q,m_{\perp}) \in \W_t(w)} \lrparen{-\beta_{t,s}^w(\Q,m_{\perp},\Q,0) + \Et{\trans{w_s^T(\Q,w_t^s(\Q,w+m_{\perp}))}(-X)}{t}}\\
&\leq \esssup_{(\Q,m_{\perp}) \in \W_{t,s}(w)} \esssup_{(\R,n_{\perp}) \in \W_s(w_t^s(\Q,w+m_{\perp}))} \lparen{-\beta_{t,s}^w(\Q,m_{\perp},\R,n_{\perp})}\\
&\quad\quad \rparen{+ \Et{\trans{w_s^T(\R,w_t^s(\Q,w+m_{\perp})+n_{\perp})}(-X)}{t}}.
\end{align*}
Now, to demonstrate $\geq$, we will show the inequality for the expectations.  Let $(\Q,m_{\perp}) \in \W_{t,s}(w)$ and let $(\R,n_{\perp}) \in \W_s(w_t^s(\Q,w+m_{\perp}))$.  It follows that
\begin{align}
\nonumber &\E{-\beta_{t,s}^w(\Q,m_{\perp},\R,n_{\perp}) + \Et{\trans{w_s^T(\R,w_t^s(\Q,w+m_{\perp})+n_{\perp})}(-X)}{t}}\\
\nonumber &= \E{\essinf_{Z \in A_t} \Et{\trans{w_s^T(\R,w_t^s(\Q,w+m_{\perp})+n_{\perp})}Z}{t}}\\
\nonumber &\qquad + \E{\Et{\trans{w_s^T(\R,w_t^s(\Q,w+m_{\perp})+n_{\perp})}(-X)}{t}}\\
\label{eq:dual-1} &= \inf_{Z \in A_t} \E{\trans{w_s^T(\R,w_t^s(\Q,w+m_{\perp})+n_{\perp})}(Z-X)}\\
\label{eq:dual-2} &\leq \inf_{u \in R_t(X)} \E{\trans{w_s^T(\R,w_t^s(\Q,w+m_{\perp})+n_{\perp})}u}\\
\nonumber &= \inf_{u \in R_t(X)} \E{\transp{w_t^s(\Q,w+m_{\perp})+n_{\perp}}\ERt{u}{s}}\\
\nonumber &= \inf_{u \in R_t(X)} \E{\transp{w+m_{\perp}}\EQt{u}{t}}\\
\label{eq:dual-3} & = \inf_{u \in R_t(X)} \E{\trans{w}u} = \E{\rho_t^w(X)}.
\end{align}
The inequality in \eqref{eq:dual-2} follows from the primal representation of the set-valued risk measure $R_t$.  The remaining lines follow from $\ERt{u}{s} = \EQt{u}{t} = u$ and $\E{\trans{m_{\perp}}u} = \E{\trans{n_{\perp}}u} = 0$.  It remains to show that we are able to interchange the expectation and the infimum above in \eqref{eq:dual-1} and \eqref{eq:dual-3} due to the $\Ft{t}$-decomposability of $A_t$.  For the terminal interchange in \eqref{eq:dual-3} this is trivial.  For \eqref{eq:dual-1} we demonstrate this by showing that $\lrcurly{\Et{\trans{w_s^T(\R,w_t^s(\Q,w+m_{\perp})+n_{\perp})}Z}{t} \; | \; Z \in A_t}$ is $\Ft{t}$-decomposable.

Let $u_X,u_Y \in \lrcurly{\Et{\trans{w_s^T(\R,w_t^s(\Q,w+m_{\perp})+n_{\perp})}Z}{t} \; | \; Z \in A_t}$ such that $u_X = \Et{\trans{w_s^T(\R,w_t^s(\Q,w+m_{\perp})+n_{\perp})}X}{t}$ and $u_Y = \Et{\trans{w_s^T(\R,w_t^s(\Q,w+m_{\perp})+n_{\perp})}Y}{t}$.  Let $D \in \Ft{t}$, then $1_D X + 1_{D^c} Y \in A_t$, and thus
\begin{align*}
1_D u_X + 1_{D^c} u_Y &= \Et{\trans{w_s^T(\R,w_t^s(\Q,w+m_{\perp})+n_{\perp})}(1_D X + 1_{D^c}Y)}{t}\\
&\in \lrcurly{\Et{\trans{w_s^T(\R,w_t^s(\Q,w+m_{\perp})+n_{\perp})}Z}{t} \; | \; Z \in A_t}.
\end{align*}

To complete the proof we will demonstrate that \eqref{eq_extended_dual} is equivalent to \eqref{eq_extended_dual-2} in much the same manner.
Since $(\Q,m_{\perp},\R,n_{\perp}) \in \tildeW_{t,s}(w)$ implies that $(\Q,m_{\perp}) \in \W_{t,s}(w)$ and $(\R,n_{\perp}) \in \W_s(w_t^s(\Q,w+m_{\perp}))$, then it immediately follows that 
\[\rho_t^w(X) \geq \esssup_{(\Q,m_{\perp},\R,n_{\perp}) \in \tildeW_{t,s}(w)} \lrparen{\begin{array}{l}  -\beta_{t,s}^w(\Q,m_{\perp},\R,n_{\perp})\\ \qquad + \Et{\trans{w_s^T(\R,w_t^s(\Q,w+m_{\perp})+n_{\perp})}(-X)}{t} \end{array}}\] 
by~\eqref{eq_extended_dual}.  Therefore it suffices to show the equivalence of the expectations of \eqref{eq_extended_dual} and \eqref{eq_extended_dual-2} in order to prove the desired property.  Beginning with~\eqref{eq_extended_dual}:
\begin{align*}
&\E{\rho_t^w(X)} = \E{\begin{array}{r}\esssup_{(\Q,m_{\perp}) \in \W_{t,s}(w)} \esssup_{(\R,n_{\perp}) \in \W_s(w_t^s(\Q,w+m_{\perp}))} \lparen{-\beta_{t,s}^w(\Q,m_{\perp},\R,n_{\perp})}\\ \rparen{+ \Et{\trans{w_s^T(\R,w_t^s(\Q,w+m_{\perp})+n_{\perp})}(-X)}{t}}\end{array}}\\
&= \sup_{(\Q,m_{\perp}) \in \W_{t,s}(w)} \sup_{(\R,n_{\perp}) \in \W_s(w_t^s(\Q,w+m_{\perp}))} \E{\begin{array}{l}-\beta_{t,s}^w(\Q,m_{\perp},\R,n_{\perp})\\ + \Et{\trans{w_s^T(\R,w_t^s(\Q,w+m_{\perp})+n_{\perp})}(-X)}{t}\end{array}}\\
&= \sup_{(\Q,m_{\perp}) \in \W_{t,s}(w)} \sup_{(\R,n_{\perp}) \in \W_s(w_t^s(\Q,w+m_{\perp}))} \lrparen{\begin{array}{l}\E{-\beta_{t,s}^w(\Q,m_{\perp},\R,n_{\perp})}\\ + \E{\trans{w_s^T(\R,w_t^s(\Q,w+m_{\perp})+n_{\perp})}(-X)}\end{array}}\\
&= \sup_{(\Q,m_{\perp},\R,n_{\perp}) \in \tildeW_{t,s}(w)} \E{\begin{array}{l}-\beta_{t,s}^w(\Q,m_{\perp},\R,n_{\perp})\\ + \Et{\trans{w_s^T(\R,w_t^s(\Q,w+m_{\perp})+n_{\perp})}(-X)}{t}\end{array}}\\
&= \E{\esssup_{(\Q,m_{\perp},\R,n_{\perp}) \in \tildeW_{t,s}(w)} \lrparen{\begin{array}{l}-\beta_{t,s}^w(\Q,m_{\perp},\R,n_{\perp})\\ + \Et{\trans{w_s^T(\R,w_t^s(\Q,w+m_{\perp})+n_{\perp})}(-X)}{t}\end{array}}}.
\end{align*}
The second and last equalities follow from $\Ft{t}$-decomposable.
\end{proof}

\begin{remark}\label{rem:extended_dual}
If all assets are eligible, $M = \R^d$, then~\eqref{eq_extended_dual} and~\eqref{eq_extended_dual-2} follow immediately and are trivially equivalent to the duality result in Proposition~\ref{prop_scalar-cc} by considering the probability measure $\S \in \mathcal{M}^d$ such that $\frac{d\S_i}{d\P} = \bar\xi_{0,s}(\Q_i)\bar\xi_{s,T}(\R_i)$ for all $i$ and, conversely, $\dQdP = \xi_{0,s}(\S)$ and $\dRdP = \xi_{s,T}(\S)$.
\end{remark}

\section{Proof of Lemma~\ref{lemma_scalar}}
\label{sec:App5}

\begin{proof}[Proof of Lemma~\ref{lemma_scalar}.]
We will prove this result by demonstrating that \eqref{thm_scalar-penalty} $\Rightarrow$ \eqref{thm_scalar-recursive} $\Rightarrow$ \eqref{thm_scalar-acceptance} $\Rightarrow$ \eqref{thm_scalar-penalty}.  Following this we will provide notes on any differences in the proof that are utilized when considering the opposite orderings that are stated in parentheses in the three assertions of Lemma~\ref{lemma_scalar}.
\begin{description}
\item[\eqref{thm_scalar-penalty} $\Rightarrow$ \eqref{thm_scalar-recursive}:]
First, consider the representation of $\rho_t^w$ for any $w \in \LdpK{q}{t}{\tilde R_t(0)^+ \backslash \prp{M}}$ given by~\eqref{eq_extended_dual}.  For any $X \in \LdpF{}$:
\begin{align*}
&\rho_t^w(X) = \esssup_{(\Q,m_{\perp}) \in \W_{t,s}(w)} \esssup_{(\R,n_{\perp}) \in \W_s(w_t^s(\Q,w+m_{\perp}))} \lparen{-\beta_{t,s}^w(\Q,m_{\perp},\R,n_{\perp})}\\
&\qquad \rparen{+ \Et{\trans{w_s^T(\R,w_t^s(\Q,w+m_{\perp})+n_{\perp})}(-X)}{t}}\\
&\geq \esssup_{\substack{(\Q,m_{\perp}) \in \W_{t,s}(w)\\ w_t^s(\Q,w+m_{\perp}) \in R_s(0)^+}} \esssup_{(\R,n_{\perp}) \in \W_s(w_t^s(\Q,w+m_{\perp}))} \lparen{-\beta_{t,s}^w(\Q,m_{\perp},\R,n_{\perp})}\\
&\qquad \rparen{+ \Et{\trans{w_s^T(\R,w_t^s(\Q,w+m_{\perp})+n_{\perp})}(-X)}{t}}\\
&\geq \esssup_{\substack{(\Q,m_{\perp}) \in \W_{t,s}(w)\\ w_t^s(\Q,w+m_{\perp}) \in R_s(0)^+}} \esssup_{(\R,n_{\perp}) \in \W_s(w_t^s(\Q,w+m_{\perp}))}
\lparen{-\alpha_{t,s}(\Q,w+m_{\perp})}\\
&\qquad \rparen{+ \Et{-\alpha_s(\R,w_t^s(\Q,w+m_{\perp})+n_{\perp}) + \transp{w_t^s(\Q,m_{\perp})+n_{\perp}}\ERt{-X}{s}}{t}}\\
&= \esssup_{\substack{(\Q,m_{\perp}) \in \W_{t,s}(w)\\ w_t^s(\Q,w+m_{\perp}) \in R_s(0)^+}} \lparen{-\alpha_{t,s}(\Q,w+m_{\perp})}\\
&\qquad\rparen{+ \Et{\esssup_{(\R,n_{\perp}) \in \W_s(w_t^s(\Q,w+m_{\perp}))} \lrparen{\begin{array}{l}-\alpha_s(\R,w_t^s(\Q,w+m_{\perp})+n_{\perp})\\ + \transp{w_t^s(\Q,w+m_{\perp})+n_{\perp}}\ERt{-X}{s}\end{array}}}{t}}\\
&= \esssup_{\substack{(\Q,m_{\perp}) \in \W_{t,s}(w)\\ w_t^s(\Q,w+m_{\perp}) \in R_s(0)^+}} \lrparen{-\alpha_{t,s}(\Q,w+m_{\perp}) +
\Et{\rho_s^{w_t^s(\Q,w+m_{\perp})}(X)}{t}}.
\end{align*}
We are able to move the essential supremum inside the conditional expectation since
\[\lrcurly{\begin{array}{l} -\alpha_s(\R,w_t^s(\Q,w+m_{\perp})+n_{\perp})\\ \quad + \transp{w_t^s(\Q,w+m_{\perp})+n_{\perp}}\ERt{-X}{s} \end{array} \; | \; (\R,n_{\perp}) \in \W_s(w_t^s(\Q,w+m_{\perp}))}\] 
is $\Ft{s}$-decomposable.

\item[\eqref{thm_scalar-recursive} $\Rightarrow$ \eqref{thm_scalar-acceptance}:]
Let $X \in A_t$ and assume $X \not\in A_{t,s} + A_s$ (which is closed and convex by Assumption~\ref{ass:convex}).  Then there exists a $Y \in \LdqF{}$ such that $\E{\trans{Y}X} < \inf_{Z \in A_{t,s} + A_s} \E{\trans{Y}Z}$.  By the monotonicity property of acceptance sets, it follows that $Y \in \LdpK{q}{}{\R^d_+}$ and, in fact, there exists $(\Q,w+m_{\perp}) \in \W_t \subseteq \W_{t,s}$ such that $Y = w_t^T(\Q,w+m_{\perp})$.  
By the appropriate decomposability of the acceptance sets, we find that
\begin{align*}
\inf_{Z \in A_{t,s} + A_s} &\E{\trans{w_t^T(\Q,w+m_{\perp})}Z} \\
&= \E{-\alpha_{t,s}(\Q,w+m_{\perp})} + \E{-\alpha_s(\Q,w_t^s(\Q,w+m_{\perp}))}.
\end{align*}
In particular, this implies that either $w \in R_t(0)^+$ or $\E{-\alpha_{t,s}(\Q,w+m_{\perp})} \leq \inf_{u \in R_t(0)} \E{\trans{w}u} = -\infty$ which is a contradiction to the strict inequality with $\E{\trans{w_t^T(\Q,w+m_{\perp})}X}$ constructed initially.  Similarly, we can conclude that $w_t^s(\Q,w+m_{\perp}) \in R_s(0)^+$.
Then it follows that
\begin{align*}
0 &= \E{\alpha_{t,s}(\Q,w+m_{\perp}) + \alpha_s(\Q,w_t^s(\Q,w+m_{\perp}))}\\
&\qquad - \E{\alpha_{t,s}(\Q,w+m_{\perp}) + \alpha_s(\Q,w_t^s(\Q,w+m_{\perp}))}\\
&< -\E{\transp{w+m_{\perp}}\EQt{X}{t}} - \E{\alpha_{t,s}(\Q,w+m_{\perp}) + \alpha_s(\Q,w_t^s(\Q,w+m_{\perp}))}\\
&= \E{-\alpha_s(\Q,w_t^s(\Q,w+m_{\perp})) + \trans{w_t^s(\Q,w+m_{\perp})}\EQt{-X}{s}}\\ 
&\qquad - \E{\alpha_{t,s}(\Q,w+m_{\perp})}\\
&\leq \E{\Et{\rho_s^{w_t^s(\Q,w+m_{\perp})}(X)}{t} - \alpha_{t,s}(\Q,w+m_{\perp})}\\
&\leq \E{\rho_t^w(X)} \leq 0.
\end{align*}
This produces a contradiction, and thus $X \in A_{t,s} + A_s$.  Note that the last step uses $X \in A_t$ if
and only if $\rho_t^w(X) \leq 0$ for every $w \in \LdpK{q}{t}{\tilde R_t(0)^+ \backslash \prp{M}}$.  

\item[\eqref{thm_scalar-acceptance} $\Rightarrow$ \eqref{thm_scalar-penalty}:]
Let $w \in \LdpK{q}{t}{\tilde R_t(0)^+ \backslash \prp{M}}$, 
$(\Q,m_{\perp}) \in \W_{t,s}(w)$, and $(\R,n_{\perp}) \in \W_s(w_t^s(\Q,w+m_{\perp}))$
\begin{align*}
\beta_{t,s}^w(\Q,m_{\perp},\R,n_{\perp}) &= \esssup_{Z_t \in A_t} \Et{\trans{w_s^T(\R,w_t^s(\Q,w+m_{\perp})+n_{\perp})}(-Z_t)}{t}\\
&\leq \esssup_{Z_{t,s} \in A_{t,s}} \esssup_{Z_s \in A_s} \Et{\trans{w_s^T(\R,w_t^s(\Q,w+m_{\perp})+n_{\perp})}(-Z_{t,s}-Z_s)}{t}\\
&= \esssup_{Z_s \in A_{t,s}} \Et{\trans{w_s^T(\R,w_t^s(\Q,w+m_{\perp})+n_{\perp})}(-Z_{t,s})}{t}\\
&\quad\quad + \esssup_{Z_s \in A_s} \Et{\trans{w_s^T(\R,w_t^s(\Q,w+m_{\perp})+n_{\perp})}(-Z_s)}{t}\\
&= \esssup_{Z_{t,s} \in A_{t,s}} \transp{w+m_{\perp}}\EQt{-Z_{t,s}}{t}\\
&\quad\quad + \esssup_{Z_s \in A_s} \Et{\transp{w_t^s(\Q,w+m_{\perp})+n_{\perp}}\ERt{Z_s}{s}}{t}\\
&= \alpha_{t,s}(\Q,w+m_{\perp}) + \Et{\alpha_s(\R,w_t^s(\Q,w+m_{\perp})+n_{\perp})}{t}
\end{align*}
Note that $\prp{M_s} = \LdpK{q}{s}{\prp{M}}$ for all times $s$, thus $\trans{n_{\perp}}Z_{t,s} = 0$ almost surely for every $Z_{t,s} \in M_s$.  Additionally, in the last line we can interchange the essential supremum and the conditional expectation by the $\Ft{s}$-decomposability of $A_s$.
\end{description}
For the opposite orderings: 
\begin{itemize}
\item In the implication from \eqref{thm_scalar-penalty} to \eqref{thm_scalar-recursive} we take advantage of 
the inequality on the penalty functions 
\[\E{\beta_{t,s}^w(\Q,m_{\perp},\R,n_{\perp})} \geq \E{\alpha_{t,s}(\Q,w+m_{\perp}) + \alpha_s(\R,w_t^s(\Q,w+m_{\perp})+n_{\perp})}\] 
and that, by definition,
\[\E{\alpha_s(\R,w_t^s(\Q,w+m_{\perp})+n_{\perp})} \geq \sup_{u \in R_s(0)} \E{\trans{w_t^s(\Q,w+m_{\perp})}(-u)} = +\infty\] 
if $w_t^s(\Q,w+m_{\perp}) \not\in R_s(0)^+$.
Thus by considering both~\eqref{eq_extended_dual} and~\eqref{eq_extended_dual-2}, we find
\begin{align*}
\rho_t^w(X) &= \esssup_{\substack{(\Q,m_{\perp}) \in \W_{t,s}(w)\\ w_t^s(\Q,w+m_{\perp}) \in R_s(0)^+}} \esssup_{(\R,n_{\perp}) \in \W_s(w_t^s(\Q,w+m_{\perp}))} \lparen{-\beta_{t,s}^w(\Q,m_{\perp},\R,n_{\perp})}\\
&\qquad\qquad \rparen{+ \Et{\trans{w_s^T(\R,w_t^s(\Q,w+m_{\perp})+n_{\perp})}(-X)}{t}}
\end{align*}
for every $X \in \LdpF{}$ and $w \in \LdpK{q}{t}{\tilde R_t(0)^+ \backslash \prp{M}}$.

\item In the implication from \eqref{thm_scalar-recursive} to \eqref{thm_scalar-acceptance} we take advantage of Lemma~\ref{prop_mptc_acceptance}.  This implies that $0 \geq -\alpha_{t,s}(\Q,w+m_{\perp}) + \Et{\rho_s^{w_t^s(\Q,w+m_{\perp})}(X)}{t}$ for every $X \in A_{t,s} + A_s$, $(\Q,m_{\perp}) \in \W_{t,s}(w)$, $w \in \LdpK{q}{t}{\tilde R_t(0)^+ \backslash \prp{M}}$, and $w_t^s(\Q,w+m_{\perp}) \in R_s(0)^+$.  Thus, it follows that 
\[0 \geq \esssup_{\substack{(\Q,m_{\perp}) \in \W_{t,s}(w)\\ w_t^s(\Q,w+m_{\perp}) \in R_s(0)^+}} \lrparen{-\alpha_{t,s}(\Q,w+m_{\perp}) + \Et{\rho_s^{w_t^s(\Q,w+m_{\perp})}(X)}{t}} \geq \rho_t^w(X)\] 
for every $w \in \LdpK{q}{t}{\tilde R_t(0)^+ \backslash \prp{M}}$.  And this is true if and only if $X \in A_t$.
\end{itemize}
\end{proof}

\begin{lemma}
\label{prop_mptc_acceptance}
Let $0 \leq t < s \leq T$.  If $X \in A_{t,s} + A_s$ then $-\rho_s^{w_s}(X) \geq \essinf_{Z \in A_{t,s}} \trans{w_s}Z$ for every $w_s \in \LdpK{q}{s}{M_+^+ \backslash \prp{M}}$.
\end{lemma}
\begin{proof}
Assume $X \in A_{t,s} + A_s$.  Then define $X_{t,s} \in A_{t,s}$ and $X_s \in A_s$ such that $X = X_{t,s} + X_s$.  It immediately follows that
\begin{align*}
-\rho_s^{w_s}(X) &= -\rho_s^{w_s}(X_s) + \trans{w_s}X_{t,s} \geq \trans{w_s}X_{t,s} \geq \essinf_{Z \in A_{t,s}} \trans{w_s}Z.
\end{align*}
\end{proof}

\bibliographystyle{plain}
\bibliography{biblio}
\end{document}